\pdfoutput=1
\documentclass[runningheads]{llncs}

\usepackage[T1]{fontenc}
\usepackage[utf8]{inputenc}
\usepackage{amsmath,amssymb}
\usepackage{booktabs}
\usepackage{microtype}
\DisableLigatures{encoding = T1, family = tt*}
\usepackage[hidelinks]{hyperref}

\hypersetup{
  hypertexnames=false,
  pdftitle={EFX Allocations for Three Agents and Seven or Eight Chores},
  pdfauthor={Xinkai Zhang},
  pdfkeywords={fair division; EFX; indivisible chores; computer-assisted proof}
}

\DeclareMathOperator{\Bad}{Bad}
\DeclareMathOperator*{\argmin}{arg\,min}
\newcommand{\R}{\mathbb{R}}
\newcommand{\lexle}{\leq_{\mathrm{lex}}}
\newcommand{\file}[1]{\path{#1}}
\expandafter\def\expandafter\UrlBreaks\expandafter{\UrlBreaks\do\-}

\begin{document}

\title{EFX Allocations for Three Agents and Seven or Eight Chores%
  \thanks{Code: \url{https://github.com/KaixxxZhang/Chores-EFX-n3m7or8}}}
\titlerunning{EFX for Three Agents and Seven or Eight Chores}

\author{Xinkai Zhang}
\authorrunning{X. Zhang}
\institute{Renmin University of China, Beijing, China\\
  \email{zhangxinkai413@ruc.edu.cn}}

\maketitle

\begin{abstract}
We show that every additive chore instance with three agents and seven or
eight indivisible chores has an EFX allocation. Costs may be arbitrary
nonnegative reals. Together with the theorem of Kobayashi, Mahara, and Sakamoto
for at most $2n$ chores, this settles the existence of EFX allocations for three
agents and at most eight chores; for four or more agents, He and Tao have shown
that EFX allocations of additive chores need not exist. Both proofs are
computer-assisted. For seven chores, we normalize each agent's costs and sort
the chores, and then express the nonexistence of an EFX allocation as a formula
of quantifier-free linear real arithmetic with one clause per allocation; the
SMT solvers Z3 and cvc5 show that this formula is
unsatisfiable. For eight chores, if two agents share a cheapest chore, we remove
that chore, apply the seven-chore result, and put the chore back using an
insertion lemma of Kobayashi, Mahara, and Sakamoto. The remaining instances, in
which the agents' sets of cheapest chores are pairwise disjoint, are handled by
a second formula of the same kind.

\end{abstract}

\section{Introduction}
\label{sec:introduction}

Envy-freeness up to any item (EFX) is one of the central fairness notions for
dividing indivisible items \cite{CaragiannisEtAl2019}. For goods, an agent must
not envy another bundle once any single good is removed from that bundle. For
chores, the removal happens on the other side: an agent must not envy another
bundle once any single chore is removed from its own bundle. This paper is about
the chores version.

For additive costs, Kobayashi, Mahara, and Sakamoto proved that an EFX
allocation exists whenever the number $m$ of chores is at most twice the number
$n$ of agents \cite{KobayashiMaharaSakamoto2025}. In general, however, EFX
allocations of chores need not exist. Christoforidis and Santorinaios gave a
counterexample with superadditive costs \cite{ChristoforidisSantorinaios2024},
and He and Tao gave counterexamples with additive costs, taking only three
values, for every $n\geq4$ \cite{HeTao2026}. For three agents with additive
costs, the $m\leq2n$ theorem covers up to six chores, and seven chores was the
smallest open case.

\subsection{Results}
\label{subsec:results}

We settle the cases of seven and eight chores. Agent $i$ has a cost
$c_{ig}\geq0$ for each chore $g$, and the cost of a set $S$ of chores is
$c_i(S)=\sum_{g\in S}c_{ig}$.

\begin{theorem}
\label{thm:seven}
Let three agents have nonnegative additive costs over seven indivisible
chores. Then there is an allocation $X=(X_0,X_1,X_2)$ of the chores such that,
for every agent $i$, either $X_i$ is empty or
\[
  c_i(X_i\setminus\{g\})\leq c_i(X_j)
  \quad\text{for every chore } g\in X_i \text{ and every agent } j.
\]
\end{theorem}

\begin{theorem}
\label{thm:eight}
The conclusion of Theorem~\ref{thm:seven} also holds for three agents and eight
chores.
\end{theorem}

Together with the $m\leq2n$ theorem \cite{KobayashiMaharaSakamoto2025}, which
covers up to six chores for three agents, this shows that every additive chore
instance with three agents and at most eight chores has an EFX allocation. As
explained in Section~\ref{sec:preliminaries}, the condition in
Theorem~\ref{thm:seven} is the usual notion of EFX for additive chores.

\subsection{Approach}
\label{subsec:approach}

We describe an instance by its $3\times m$ cost matrix $C=(c_{ig})$: row $i$
lists the costs of agent $i$, and column $g$ lists the costs of chore $g$ for
the three agents. Informally, $C$ is a counterexample if none of its $3^m$
allocations is EFX (Definition~\ref{def:counterexample}). For a fixed
allocation, the failure of EFX is a disjunction of strict linear inequalities in
the entries of $C$. A counterexample is therefore a
solution of the conjunction of these disjunctions over all $3^m$ allocations,
which is a formula of quantifier-free linear real arithmetic (QF\_LRA). We show
with the SMT solvers Z3 \cite{deMouraBjorner2008} and cvc5
\cite{BarbosaEtAl2022} that this formula, restricted to a suitable canonical
form, has no solution.

For seven chores the canonical form comes from three elementary observations.
An agent whose costs are all zero is handled by an explicit allocation.
Otherwise each row can be scaled to sum to one, and the columns can be sorted,
since neither operation changes which allocations are EFX. The resulting formula
has 21 real variables and one clause for each of the $2187$ allocations, and
both solvers show that it is unsatisfiable. Restricting the formula to
allocations with one combination of bundle sizes would not suffice:
Appendix~\ref{app:occupancy} gives an instance in which every EFX allocation has
bundle sizes $(3,3,1)$, and similar instances for $(4,2,1)$ and $(5,1,1)$, so an
argument that only considers allocations with bundle sizes $(3,2,2)$ would miss
them.

For eight chores, the same canonical form gives a formula with $6561$ clauses
that timed out in several solver configurations, so we first treat part of the
instances by hand. If some chore is a cheapest chore for two agents, we delete
it, apply Theorem~\ref{thm:seven} to the remaining seven chores, and put the
chore back using an insertion lemma of Kobayashi, Mahara, and Sakamoto
\cite{KobayashiMaharaSakamotoArXiv}. In all remaining instances the three
agents' sets of cheapest chores are pairwise disjoint. We can then place a
cheapest chore of agent $i$ in column $i$, for each $i$, and the formula for
this case, with 24 variables and $6561$ clauses, is again unsatisfiable.

\subsection{Related Work}
\label{subsec:related}

\paragraph{EFX for additive chores.}
Kobayashi, Mahara, and Sakamoto \cite{KobayashiMaharaSakamoto2025} introduced
the EFX graph, in which a perfect matching assigns bundles to agents so that the
result is EFX. With it they compute EFX allocations in polynomial time in three
cases: when $m\leq2n$, when all agents but one order the chores in the same way,
and for three agents with personalized bi-valued costs. For $n=3$ their first
result ends at $m=6$. In Section~\ref{sec:eight} we use their Lemma~4.2
and Observations~2.2 and~2.3 \cite{KobayashiMaharaSakamotoArXiv}. Their second
case extends the case of identically ordered agents, solved by Li, Li, and Wu
\cite{LiLiWu2022}, who turned the top-trading envy-cycle elimination algorithm
of Bhaskar, Sricharan, and Vaish \cite{BhaskarSricharanVaish2021} from an EF1
algorithm into an EFX algorithm by using the common order. Yin and Mehta
\cite{YinMehta2022} proved existence for three agents when two of them are
additive, order the chores in the same way, and find every bundle of two or
more chores costlier than any single chore.

Garg, Murhekar, and Qin \cite{GargMurhekarQin2025} gave another proof of the
case $m\leq2n$. It starts from bundles of at most two chores and repairs
violations by swapping chores. For $n=3$ and $m=7$ some bundle must contain at
least three chores, so this starting point is not available. For general
instances they obtain a constant-factor approximation of EFX; the best factor is
now $2$, due to Garg and Murhekar \cite{GargMurhekar2026}. For three agents,
$2$-EFX allocations exist for all subadditive costs, and tEFX allocations, in
which envy disappears after any one chore is transferred to the envied agent,
exist under an additional bound on cost ratios \cite{AfshinmehrEtAl2024}.

\paragraph{Goods, and computer-assisted proofs.}
For goods, Chaudhury, Garg, and Mehlhorn \cite{ChaudhuryGargMehlhorn2024}
proved that EFX allocations exist for three additive agents. Their
champion-graph and trimming arguments remove goods from the envied bundle and do
not carry over directly to chores.

Alkassar, Fouz, and Mehlhorn \cite{AlkassarFouzMehlhorn2026} proved that
complete EFX allocations exist for four additive agents and at most nine goods.
Like us, they combine reductions by hand with QF\_LRA unsatisfiability checks
over a continuous domain. They split the space of valuations into $36{,}152$
canonical regions and rule out each region with its own list of candidate
allocations, so their proof also depends on the regions covering all cases. In
our proofs each theorem is decided by a single formula over the whole canonical
domain, with one clause for every allocation.

Akrami, Mayorov, Mehlhorn, Srinivas, and Weidenbach \cite{AkramiEtAl2026} used
SAT solvers to construct counterexamples to EFX for monotone and for submodular
valuations of goods, and to prove that every instance with three agents,
monotone valuations and seven goods has an EFX allocation. These valuations are
not additive, so their results do not apply here. Their positive result is
encoded through Boolean comparisons between subsets, and the encoding is
formally verified in Lean. Section~10 of their paper is closest to our approach:
it expresses nonexistence as a QF\_LRA formula with one clause per allocation
and gives it to Z3. Monotone valuations have no compact description, so their
variables are indexed by pairs of an agent and a set of goods, $3\cdot2^7$ of
them for $m=7$ compared with $21$ here, and their only symmetry reduction fixes
one agent's order on single goods. They report $m=6$ in about $25$ seconds and
no answer for $m=7$ within $40$ hours. Most of this difference comes from
additivity rather than from our reductions.

\paragraph{Organization.}
Section~\ref{sec:preliminaries} defines EFX for chores and counterexamples.
Sections~\ref{sec:seven} and~\ref{sec:eight} prove Theorems~\ref{thm:seven}
and~\ref{thm:eight}. Section~\ref{sec:computation} describes the solver runs and
the checks on the encodings, and Section~\ref{sec:discussion} discusses the next
case. Appendix~\ref{app:occupancy} contains the examples about bundle sizes
mentioned above.

\section{Preliminaries}
\label{sec:preliminaries}

Let $N=\{0,1,2\}$ be the set of agents and $M$ a set of $m$ chores. Agent $i$
has a cost $c_{ig}\in\R_{\geq0}$ for each chore $g$, and costs are additive:
$c_i(S)=\sum_{g\in S}c_{ig}$ for $S\subseteq M$, so $c_i(\varnothing)=0$. As in
Section~\ref{subsec:approach}, the costs form a $3\times m$ matrix $C=(c_{ig})$
with one row per agent and one column per chore; we write
$K_g=(c_{0g},c_{1g},c_{2g})$ for column $g$. An allocation
$X=(X_0,X_1,X_2)$ is a partition of $M$ into three bundles, some of which may be
empty; agent $i$ receives $X_i$.

\begin{definition}[EFX for chores]
\label{def:efx}
An allocation $X$ is \emph{EFX} if, for every agent $i\in N$,
\begin{equation}
  X_i=\varnothing
  \quad\text{or}\quad
  c_i(X_i\setminus\{g\})\leq c_i(X_j)
  \text{ for every } g\in X_i \text{ and every } j\in N.
  \label{eq:efx}
\end{equation}
\end{definition}

In \eqref{eq:efx} agent $i$ evaluates both bundles with its own costs, that is,
with row $i$. The condition is required for every chore $g\in X_i$, including
chores with $c_{ig}=0$.

\begin{remark}[Residual form]
\label{rem:residual}
For nonempty $X_i$, additivity gives
$\max_{g\in X_i}c_i(X_i\setminus\{g\})=c_i(X_i)-\min_{g\in X_i}c_{ig}$.
Hence \eqref{eq:efx} is equivalent to
\begin{equation}
  \tau_i(X)\leq c_i(X_j)\quad\text{for all } i,j\in N,
  \label{eq:residual}
\end{equation}
where $\tau_i(X)=0$ if $X_i=\varnothing$ and
$\tau_i(X)=c_i(X_i)-\min_{g\in X_i}c_{ig}$ otherwise. This is the form in which
EFX for additive chores is usually stated, for instance in the EFX graph of
\cite{KobayashiMaharaSakamoto2025}. It treats zero-cost chores in the same way
as Definition~\ref{def:efx}: if
$\min_{g\in X_i}c_{ig}=0$, the condition becomes $c_i(X_i)\leq c_i(X_j)$. So
Definition~\ref{def:efx} agrees with the usual definition, and the results
cited in Section~\ref{sec:introduction} concern the same notion.
\end{remark}

Label the chores $0,\ldots,m-1$. An allocation corresponds to a vector
$a=(a_0,\ldots,a_{m-1})\in N^m$, where $a_g$ is the agent who receives chore
$g$; we write $X_k(a)=\{g\in M:a_g=k\}$. Negating \eqref{eq:efx}, the allocation
$a$ is not EFX exactly when there are an agent $i$, a chore $g\in X_i(a)$, and an
agent $j$ with
\begin{equation}
  \sum_{h\in X_i(a)\setminus\{g\}}c_{ih}
  >
  \sum_{h\in X_j(a)}c_{ih}.
  \label{eq:strict-failure}
\end{equation}
For $j=i$ this never holds, because removing a chore of nonnegative cost cannot
increase the cost of $X_i(a)$. We therefore leave out $j=i$ in the formulas
below.

\begin{definition}[Counterexample]
\label{def:counterexample}
A nonnegative $3\times m$ cost matrix is a \emph{counterexample} if none of its
allocations is EFX, that is, if each of the $3^m$ vectors $a\in N^m$ satisfies
\eqref{eq:strict-failure} for some $i$, some $g\in X_i(a)$, and some $j$.
\end{definition}

\section{Seven Chores}
\label{sec:seven}

In this section $m=7$ and $M=\{0,\ldots,6\}$, except in
Lemma~\ref{lem:zero-row}, which we state for every $m\geq2$.

\subsection{Canonical Form}
\label{subsec:canonical}

\begin{lemma}[Zero row]
\label{lem:zero-row}
Let $m\geq2$. If some agent $z$ has $c_{zg}=0$ for every chore $g$, then the
instance has an EFX allocation.
\end{lemma}

\begin{proof}
Give one chore to each of the two other agents and the remaining $m-2$ chores to
$z$. Each agent other than $z$ owns a single chore; removing it leaves cost
zero, which is at most the cost of any bundle. Agent $z$ has cost zero for every
bundle, with or without a chore removed. So \eqref{eq:efx} holds for all three
agents. \qed
\end{proof}

\begin{lemma}[Row scaling]
\label{lem:row-scaling}
For positive numbers $\lambda_0,\lambda_1,\lambda_2$, replacing each $c_{ig}$ by
$\lambda_i c_{ig}$ does not change the set of EFX allocations. Consequently, if
every row has a positive sum, we may assume that
\begin{equation}
  c_{ig}\geq0
  \quad\text{and}\quad
  \sum_{g\in M}c_{ig}=1
  \qquad\text{for every } i\in N.
  \label{eq:normalization}
\end{equation}
\end{lemma}

\begin{proof}
Each inequality in \eqref{eq:efx} compares two sums from the same row $i$.
Multiplying row $i$ by $\lambda_i>0$ multiplies both sides by $\lambda_i$ and
keeps the inequality. If the sum of row $i$ is positive, take
$\lambda_i=(\sum_g c_{ig})^{-1}$. \qed
\end{proof}

For $x,y\in\R^3$ we use the lexicographic order
\begin{equation}
  \begin{aligned}
    x\lexle y
    \iff{}& x_0<y_0
    \ \lor\ (x_0=y_0\land x_1<y_1)\\
    &\lor\ (x_0=y_0\land x_1=y_1\land x_2<y_2)
    \ \lor\ x=y.
  \end{aligned}
  \label{eq:lex}
\end{equation}
The last case allows equal columns.

\begin{lemma}[Column sorting]
\label{lem:column-sort}
Relabeling the chores does not change whether an EFX allocation exists. Hence
we may assume that
\begin{equation}
  K_0\lexle K_1\lexle\cdots\lexle K_6.
  \label{eq:column-sort}
\end{equation}
\end{lemma}

\begin{proof}
A permutation of the chores maps each allocation to the allocation obtained by
applying the permutation to all three bundles. This is a bijection between the
allocations before and after relabeling, and it preserves every sum in
\eqref{eq:efx}. Any finite list of columns can be sorted with respect to
\eqref{eq:lex}, even when some columns are equal. \qed
\end{proof}

So if some row of $C$ is zero, Lemma~\ref{lem:zero-row} applies. Otherwise all
row sums are positive, and by Lemmas~\ref{lem:row-scaling}
and~\ref{lem:column-sort} we may assume that $C$ satisfies
\eqref{eq:normalization} and \eqref{eq:column-sort}.

\subsection{The Formula \texorpdfstring{$\Phi_7$}{Phi7}}
\label{subsec:phi7}

The formula has the 21 entries $c_{ig}$ as real variables. For each $a\in N^7$ let
\begin{equation}
  \Bad_a(C)=
  \bigvee_{\substack{i\in N,\ g\in X_i(a)\\ j\in N\setminus\{i\}}}
  \Bigl(
    \sum_{h\in X_i(a)}c_{ih}-c_{ig}
    >
    \sum_{h\in X_j(a)}c_{ih}
  \Bigr).
  \label{eq:bad}
\end{equation}
An agent with an empty bundle contributes no literals, and each chore
$g\in X_i(a)$ contributes one literal for each $j\neq i$. By
\eqref{eq:strict-failure}, and since the case $j=i$ can be left out,
$\Bad_a(C)$ holds exactly when the allocation $a$ is not EFX for $C$. Let
\begin{align}
  \Phi_7(C)={}&
  \Bigl(\bigwedge_{\substack{i\in N\\g\in M}} c_{ig}\geq0\Bigr)
  \land
  \Bigl(\bigwedge_{i\in N}\sum_{g\in M}c_{ig}=1\Bigr)
  \land
  \Bigl(\bigwedge_{g=0}^{5}K_g\lexle K_{g+1}\Bigr)
  \notag\\
  &\land
  \Bigl(\bigwedge_{a\in N^7}\Bad_a(C)\Bigr).
  \label{eq:phi7}
\end{align}
The row sums restrict each row to a simplex, but by Lemma~\ref{lem:row-scaling}
this loses no counterexample. The last part of $\Phi_7$ asserts $\Bad_a$ for
every one of the $3^7=2187$ allocations $a$. Since the bundles together contain
seven chores, each of these clauses has exactly $14$ literals. With the 21
nonnegativity constraints, the three row sums and the six column comparisons,
$\Phi_7$ consists of $2217$ assertions.

\begin{proposition}
\label{prop:phi7-exact}
The formula $\Phi_7$ is satisfiable if and only if there is a counterexample
whose three row sums are positive.
\end{proposition}

\begin{proof}
Let $C$ be a counterexample with positive row sums. Scale its rows as in
Lemma~\ref{lem:row-scaling} and sort its columns as in
Lemma~\ref{lem:column-sort}. Since every allocation of $C$ fails EFX and
relabeling the chores is a bijection on allocations, every allocation of the
resulting matrix fails EFX as well. So every $\Bad_a$ holds for it, and it is a
solution of $\Phi_7$.

Conversely, a solution of $\Phi_7$ is a nonnegative matrix with unit row sums,
so its row sums are positive, and $\Bad_a$ holds for every allocation $a$. By
\eqref{eq:bad}, no allocation is EFX, so the matrix is a counterexample. \qed
\end{proof}

\subsection{Deciding \texorpdfstring{$\Phi_7$}{Phi7}}
\label{subsec:decide7}

\begin{proposition}
\label{prop:phi7-unsat}
The formula $\Phi_7$ is unsatisfiable.
\end{proposition}

\begin{proof}[computer-assisted]
The SMT solvers Z3 and cvc5 both report that $\Phi_7$ is unsatisfiable.
Section~\ref{sec:computation} describes the programs and the runs. \qed
\end{proof}

\begin{proof}[of Theorem~\ref{thm:seven}]
If some row of the cost matrix is zero, Lemma~\ref{lem:zero-row} gives an EFX
allocation. Otherwise every row sum is positive, and a counterexample would give
a solution of $\Phi_7$ by Proposition~\ref{prop:phi7-exact}, contradicting
Proposition~\ref{prop:phi7-unsat}. Hence every instance with three agents and
seven chores has an EFX allocation. \qed
\end{proof}

\section{Eight Chores}
\label{sec:eight}

In this section $m=8$ and $M=\{0,\ldots,7\}$. Lemmas~\ref{lem:zero-row},
\ref{lem:row-scaling} and~\ref{lem:column-sort} hold for any number of chores,
with the same proofs. Before building a formula, we treat by hand the instances
in which two agents share a cheapest chore.

\subsection{Shared Cheapest Chores}
\label{subsec:shared}

For each agent $i$, let
\begin{equation}
  L_i=\argmin_{g\in M}c_{ig}
  \label{eq:argmins}
\end{equation}
be the set of all cheapest chores of agent $i$.

We now recall the results of Kobayashi, Mahara, and Sakamoto that we use, with
the numbering of the full version \cite{KobayashiMaharaSakamotoArXiv}. Let
$A=(A_u)_{u\in U}$ be a partition of a set of chores into bundles indexed by a
set $U$ of $n$ bundle vertices; the bundles are not yet assigned to agents. The
\emph{EFX graph} $G_A$ is the bipartite graph with the agents on one side and $U$
on the other, in which, for nonempty $A_u$, agent $i$ is adjacent to $u$ exactly
when
$\max_{f\in A_u}c_i(A_u\setminus\{f\})\leq\min_{v\in U}c_i(A_v)$. By their
Observation~2.2(i), if $|A_u|\leq1$, including $A_u=\varnothing$, then $(i,u)$
is an edge for every agent $i$. By their Observation~2.3, a perfect matching in
$G_A$ yields an EFX allocation: each agent receives the bundle it is matched to.
The source numbers the agents $1,\ldots,n$, and we write
$[n-1]=\{1,\ldots,n-1\}$.

\begin{lemma}[Kobayashi--Mahara--Sakamoto, Lemma~4.2]
\label{lem:kms-insertion}
Let $S$ be a finite set of chores, let $M'\subsetneq S$, and let
$e\in S\setminus M'$ satisfy $c_{ie}\leq c_{ie'}$ for every $i\in[n-1]$ and
every $e'\in M'$. Let $A=(A_u)_{u\in U}$ be a partition of $M'$ into bundles
indexed by $U$ whose EFX graph $G_A$ has a perfect matching. Then there is a
vertex $u^\ast\in U$ such that the EFX graph still has a perfect matching after
$e$ is added to $A_{u^\ast}$.
\end{lemma}

In \cite{KobayashiMaharaSakamotoArXiv}, Lemma~4.2 is stated in Section~4, under
the standing assumption that agents $1,\ldots,n-1$ have identically ordered
costs. Its proof, however, uses only the displayed condition on $e$: it relies
on Observation~2.2, Lemma~2.4 and Corollary~2.5, and not on the identical
ordering. The authors note at the statement that the lemma ``will be used also
in Section~5'', and there they apply it to personalized bi-valued instances,
which need not be identically ordered \cite[Lemma~4.2; Sect.~5,
Case~1-2]{KobayashiMaharaSakamotoArXiv}. The remaining agent $n$ need not regard
$e$ as a cheapest chore.

\begin{lemma}[Shared cheapest chore]
\label{lem:shared}
If $L_p\cap L_q\neq\varnothing$ for two distinct agents $p$ and $q$, then the
eight-chore instance has an EFX allocation.
\end{lemma}

\begin{proof}
Choose $e\in L_p\cap L_q$, delete it, and apply Theorem~\ref{thm:seven} to the
remaining instance with three agents and seven chores. Let $A=(A_0,A_1,A_2)$ be
the resulting EFX allocation, now regarded as a partition into three bundle
vertices. The identity assignment, which matches agent $i$ to $A_i$, is a
perfect matching in $G_A$. Indeed, if $A_i\neq\varnothing$, the edge from $i$
to $A_i$ is exactly the residual form \eqref{eq:residual} of EFX for agent $i$,
and if $A_i=\varnothing$, the edge is supplied by Observation~2.2(i).

Now switch to the one-based labels of the source and relabel the agents so that
$\{p,q\}=[n-1]$. Relabeling permutes the agent side of $G_A$ and carries the
perfect matching along; the matching need not be the identity any more. Since
$e$ lies in both $L_p$ and $L_q$, we have $c_{ie}\leq c_{ih}$ for
$i\in\{p,q\}$ and every $h\in M\setminus\{e\}$. So all hypotheses of
Lemma~\ref{lem:kms-insertion} hold with $M'=M\setminus\{e\}$. Adding $e$ to the
bundle given by the lemma yields a partition of all eight chores whose EFX graph
has a perfect matching, and Observation~2.3 turns it into an EFX allocation
\cite{KobayashiMaharaSakamotoArXiv}. \qed
\end{proof}

\subsection{The Formula \texorpdfstring{$\Phi_8$}{Phi8}}
\label{subsec:phi8}

By Lemmas~\ref{lem:zero-row} and~\ref{lem:shared}, a counterexample has
no zero row and satisfies
\begin{equation}
  L_0\cap L_1=L_0\cap L_2=L_1\cap L_2=\varnothing.
  \label{eq:disjoint}
\end{equation}
For such a matrix, choose $e_i\in L_i$ for each agent $i$. By
\eqref{eq:disjoint} the three chores are distinct, so we can relabel the chores
so that $e_i$ becomes chore $i$, for $i=0,1,2$; as in
Lemma~\ref{lem:column-sort}, relabeling does not change whether an EFX
allocation exists. All row sums are positive, so by Lemma~\ref{lem:row-scaling}
we may again assume unit row sums. The chores $3,\ldots,7$ can still be permuted
among themselves, so we may sort columns $3,\ldots,7$ in the order
\eqref{eq:lex}.

For $a\in N^8$ we define $\Bad_a(C)$ by \eqref{eq:bad}, now with
$M=\{0,\ldots,7\}$. Again $\Bad_a(C)$ holds exactly when $a$ is not EFX for $C$,
and each clause has $2\cdot8=16$ literals.
The formula $\Phi_8$ has the 24 entries $c_{ig}$ as real variables:
\begin{align}
  \Phi_8(C)={}&
  \Bigl(\bigwedge_{\substack{i\in N\\g\in M}}c_{ig}\geq0\Bigr)
  \land
  \Bigl(\bigwedge_{i\in N}\sum_{g\in M}c_{ig}=1\Bigr)
  \land
  \Bigl(\bigwedge_{\substack{i\in N\\g\in M}}c_{ii}\leq c_{ig}\Bigr)
  \notag\\
  &\land
  \bigwedge_{\substack{p,q\in N,\ p<q\\g\in M}}
  \Bigl(
    \bigvee_{h\neq g}c_{ph}<c_{pg}
    \ \lor\
    \bigvee_{h\neq g}c_{qh}<c_{qg}
  \Bigr)
  \notag\\
  &\land
  \Bigl(\bigwedge_{g=3}^{6}K_g\lexle K_{g+1}\Bigr)
  \land
  \Bigl(\bigwedge_{a\in N^8}\Bad_a(C)\Bigr).
  \label{eq:phi8}
\end{align}
The third part says that chore $i$ is a cheapest chore of agent $i$. The fourth
part expresses \eqref{eq:disjoint}: for each pair of agents $p<q$ and each chore
$g$, some other chore is strictly cheaper than $g$ for $p$ or for $q$, so $g$ is
not a cheapest chore of both. The last part has one clause for each of the
$3^8=6561$ allocations, and $\Phi_8$ consists of $6640$ assertions in total.

\begin{proposition}
\label{prop:phi8-exact}
The formula $\Phi_8$ is satisfiable if and only if there is an eight-chore
counterexample that satisfies \eqref{eq:disjoint}.
\end{proposition}

\begin{proof}
Let $C$ be such a counterexample. Choose one cheapest chore of each agent; by
\eqref{eq:disjoint} these chores are distinct, so we can move them to columns
$0,1,2$. Then scale the rows, which have positive sums by
Lemma~\ref{lem:zero-row}, to unit sums and sort the five remaining
columns. These operations preserve the EFX allocations, up to the
relabeling of chores. The resulting matrix still satisfies \eqref{eq:disjoint},
which gives the fourth part of \eqref{eq:phi8}, and since no allocation is EFX,
every $\Bad_a$ holds. So the matrix is a solution of $\Phi_8$.

Conversely, a solution of $\Phi_8$ is a nonnegative matrix with unit row sums
whose sets of cheapest chores are pairwise disjoint, and its $6561$ allocation
clauses state that no allocation is EFX. So it is a counterexample that
satisfies \eqref{eq:disjoint}. \qed
\end{proof}

\begin{remark}
\label{rem:full-eight}
The canonical form of Section~\ref{subsec:canonical} alone gives, for eight
chores, the formula with unit row sums, all eight columns sorted and the same
$6561$ allocation clauses. This formula timed out in several solver
configurations. With Lemma~\ref{lem:shared}, the instances excluded by
\eqref{eq:disjoint} are handled by hand, and it suffices to show that $\Phi_8$
is unsatisfiable.
\end{remark}

\subsection{Deciding \texorpdfstring{$\Phi_8$}{Phi8}}
\label{subsec:decide8}

\begin{proposition}
\label{prop:phi8-unsat}
The formula $\Phi_8$ is unsatisfiable.
\end{proposition}

\begin{proof}[computer-assisted]
Z3 and cvc5 both report that $\Phi_8$ is unsatisfiable.
Section~\ref{sec:computation} describes the programs and the runs. \qed
\end{proof}

\begin{proof}[of Theorem~\ref{thm:eight}]
If some row is zero, apply Lemma~\ref{lem:zero-row}. Otherwise, if two of
the sets $L_i$ intersect, apply Lemma~\ref{lem:shared}. In the remaining case
\eqref{eq:disjoint} holds, and a counterexample would give a solution of
$\Phi_8$ by Proposition~\ref{prop:phi8-exact}, contradicting
Proposition~\ref{prop:phi8-unsat}. Hence every instance with three agents and
eight chores has an EFX allocation. \qed
\end{proof}

\section{Computations}
\label{sec:computation}

\subsection{Programs and Runs}
\label{subsec:runs}

The programs are in the repository
\url{https://github.com/KaixxxZhang/Chores-EFX-n3m7or8}. The scripts
\file{smt/seven_chores_z3.py} and \file{smt/seven_chores_cvc5.py} build
$\Phi_7$, and \file{smt/eight_chores_z3.py} and \file{smt/eight_chores_cvc5.py}
build $\Phi_8$ when called with the option \file{--disjoint-argmins}; without
this option they build the formula of Remark~\ref{rem:full-eight}. Each script
runs through all $3^m$ allocations and adds one clause for each. The variables
have the SMT sort Real, so the solvers work with exact rational numbers. The
file \file{RUNS.md} in the repository gives the command for each run in
Tables~\ref{tab:runs} and~\ref{tab:further}.

Table~\ref{tab:runs} lists the four runs behind
Propositions~\ref{prop:phi7-unsat} and~\ref{prop:phi8-unsat}. They used
\texttt{z3-solver}~5.1.0.0 and \texttt{cvc5}~1.3.4 under CPython~3.13.3 on a
14-core Apple M4 Pro. For cvc5 we set the options \texttt{produce-proofs} and
\texttt{check-proofs}, so that cvc5 checks its own proof internally. In
September 2026 the four runs were repeated on the same machine with the same
solver versions. All four again returned \texttt{unsat}, with the check times
shown in the last column of Table~\ref{tab:runs}.

\begin{table}[t]
  \centering
  \caption{The runs for Propositions~\ref{prop:phi7-unsat}
  and~\ref{prop:phi8-unsat}. Times are in seconds: \emph{build} is the time to
  construct the formula, \emph{check} the time the solver took, and
  \emph{rerun} the check time when the run was repeated in September 2026.}
  \label{tab:runs}
  \small
  \begin{tabular}{@{}lllrrr@{}}
    \toprule
    Formula & Solver & Result & Build & Check & Rerun \\
    \midrule
    $\Phi_7$ & Z3 5.1.0                    & \texttt{unsat} & 0.959 & 517.911  & 485.669 \\
    $\Phi_7$ & cvc5 1.3.4                  & \texttt{unsat} & 0.114 & 483.318  & 399.756 \\
    $\Phi_8$ & Z3 5.1.0, arith.\ solver 2  & \texttt{unsat} & 4.625 & 852.824  & 606.183 \\
    $\Phi_8$ & cvc5 1.3.4                  & \texttt{unsat} & 0.521 & 3729.851 & 3156.109 \\
    \bottomrule
  \end{tabular}
\end{table}

\begin{table}[t]
  \centering
  \caption{Further runs of the eight-chore programs. The rows with $m=7$ use
  seven chores, where Theorem~\ref{thm:seven} gives the answer. The last two rows
  change the clauses and are expected to be satisfiable.}
  \label{tab:further}
  \small
  \begin{tabular}{@{}llr@{}}
    \toprule
    Run & Result & Check (s) \\
    \midrule
    $\Phi_8$, Z3, arith.\ solver 2, phase selection 1 & \texttt{unsat} & 850.863 \\
    $\Phi_8$, Z3, arith.\ solver 6                    & \texttt{unsat} & 2717.209 \\
    $\Phi_8$, second encoding (Z3)                    & \texttt{unsat} & 469.389 \\
    $m=7$, Z3, arith.\ solver 2                       & \texttt{unsat} & 17.952 \\
    $m=7$, cvc5                                       & \texttt{unsat} & 44.497 \\
    $m=7$, second encoding (Z3)                       & \texttt{unsat} & 4.725 \\
    $\Phi_8$ with $\geq$ instead of $>$               & \texttt{sat}   & 79.500 \\
    $\Phi_8$ with no chore removed                    & \texttt{sat}   & 12.036 \\
    \bottomrule
  \end{tabular}
\end{table}

Table~\ref{tab:further} lists further runs. Z3 also decided $\Phi_8$ with phase
selection~1 instead of its default~3, and with arithmetic solver~6. The script
\file{smt/eight_chores_alt_z3.py} is a second encoding of $\Phi_8$, written
separately and sharing no code with the others. It requires only positive row
sums instead of unit row sums, sorts columns $3,\ldots,7$ in decreasing order,
and keeps the comparisons with $j=i$, so that each clause has 24 literals. The
three eight-chore programs were also run on seven chores, with the cheapest
chores pinned in the same way; this formula is unsatisfiable by
Theorem~\ref{thm:seven}, and all three programs confirm it.

\subsection{Checks on the Encodings}
\label{subsec:checks}

A proof by unsatisfiability goes wrong if the formula is stronger than intended,
for example if some clause $\Bad_a$ is stronger than the statement that $a$ is
not EFX, or if the constraints on the matrix exclude some counterexamples. The
encodings are not formally verified. Besides the second encoding of $\Phi_8$,
we made the following checks.

\paragraph{Clauses against a direct implementation.}
The repository contains an exhaustive EFX checker, written directly from
Definition~\ref{def:efx}, that goes through all allocations using exact
fractions. The script \file{smt/check_eight_chores_encoding.py} evaluates the
clauses of both Z3 encodings of $\Phi_8$ on 8 exact rational matrices and all
$6561$ allocations, $52{,}488$ pairs in total, and finds no disagreement with
the checker. The same comparison detects five plausible mistakes: removing a
chore from the envied bundle, as for goods; skipping zero-cost chores; using
$\geq$ instead of $>$; costing $X_j$ with row $j$ instead of row $i$; and
removing no chore at all. The clause function of
\file{smt/seven_chores_z3.py} is the same code as that of
\file{smt/eight_chores_z3.py}, specialized to $m=7$. The checker itself is
tested on the counterexamples of He and Tao \cite{HeTao2026}, which have no EFX
allocation, and on random instances with three agents and six chores, which have
one by the $m\leq2n$ theorem.

\paragraph{The reduction steps.}
The same script checks the construction of Lemma~\ref{lem:zero-row} for $m=8$
($3360$ cases), the invariance of the set of EFX allocations under row scaling
and under relabeling of the chores ($78{,}732$ allocations each), the canonical
forms of Sections~\ref{subsec:canonical} and~\ref{subsec:phi8} on $1000$ random
matrices each, and the steps in the proof of Lemma~\ref{lem:shared} on $30$
random instances with a shared cheapest chore.

\paragraph{Satisfiable variants.}
If $>$ is replaced by $\geq$ in every clause, or if no chore is removed, so that
the clauses express ordinary envy, then both $\Phi_7$ and $\Phi_8$ become
satisfiable. Z3 finds solutions within seconds for $\Phi_7$, and within
80 seconds for $\Phi_8$ (the last two rows of Table~\ref{tab:further}). So the
unsatisfiability of $\Phi_7$ and $\Phi_8$ is not a general effect of putting
thousands of disjunctive clauses on 21 or 24 variables.

\section{Discussion}
\label{sec:discussion}

\paragraph{The next case.}
Theorems~\ref{thm:seven} and~\ref{thm:eight} extend the existence of EFX
allocations for three agents with additive costs from six to eight chores. The
counterexamples of He and Tao \cite{HeTao2026} show that these results are not
cases of a statement for every $n$. Those counterexamples do not settle the
cases $m=2n+1$ and $m=2n+2$ for larger $n$ either: the smallest has $n=4$ and
$m=13$, and the general construction uses more than $3n$ chores.

The argument of Lemma~\ref{lem:shared} also applies to nine chores: if two
agents share a cheapest chore, delete it, apply Theorem~\ref{thm:eight} to the
remaining eight chores, and put it back with Lemma~\ref{lem:kms-insertion}.
What remains is the case of pairwise disjoint sets of cheapest chores, with a
formula of $3^9$ allocation clauses. We leave the case $m=9$ open. The number
of clauses grows exponentially with $m$, and the eight-chore formula already
needed Lemma~\ref{lem:shared} and a particular Z3 setting. Rather than deciding
one more case, it seems more useful to understand why the disjoint case always
has EFX allocations, for example through a potential-function argument that
works for every $m$.

\paragraph{Which constraints the solvers need.}
For seven chores, sorting the columns is essential. A Z3 run on $\Phi_7$ without
the column constraints was stopped after 4 hours and 19 minutes without an
answer, and with only the constraints $c_{ig}\geq0$, without normalization and
sorting, both solvers returned \texttt{unknown} after 3600 seconds. Row
normalization matters less: the second encoding of $\Phi_8$ requires only
positive row sums and was decided in 469 seconds.

\begin{credits}
\subsubsection{\ackname}
The author thanks the anonymous reviewers for their comments on an earlier
version of this paper.

\subsubsection{Use of AI tools.}
Large language models (GPT-5.6 Sol, Claude Opus 5 and 5.5, Grok 4.6 and Kimi
K3, used as agents in the Cursor editor) were used under the author's direction
to draft the arguments, the programs and the text of this paper, and to review
the drafts. The author has checked all mathematical arguments and all code, and
takes full responsibility for the content. This statement follows the
recommendations of the Leiden declaration on artificial intelligence and
mathematics \cite{LeidenDeclaration2026}.
\end{credits}

\bibliographystyle{splncs04}
\bibliography{refs}

\appendix

\section{Bundle Sizes of EFX Allocations}
\label{app:occupancy}

This appendix is not used in the proofs. It shows that a proof for seven chores
cannot restrict attention to allocations with one particular combination of
bundle sizes. The \emph{occupancy} of an allocation is the triple of its bundle
sizes in decreasing order; for example, bundles of sizes $2$, $3$ and $2$ give
occupancy $(3,2,2)$. As in Section~\ref{sec:eight}, $L_i$ denotes the set of
cheapest chores of agent $i$, now with $M=\{0,\ldots,6\}$.

\begin{lemma}
\label{lem:no-empty}
If $L_0,L_1,L_2$ are pairwise disjoint, then every EFX allocation has three
nonempty bundles.
\end{lemma}

\begin{proof}
Suppose that agent $k$ has $X_k=\varnothing$. If another agent $i$ owns at least
two chores, comparing with $X_k$ in \eqref{eq:efx} gives
$c_i(X_i\setminus\{g\})\leq0$ for every $g\in X_i$. By nonnegativity every chore
in $X_i$ then has cost zero for $i$, so $X_i\subseteq L_i$.

For occupancy $(7,0,0)$, the agent owning all seven chores has $L_i=M$, which
contradicts disjointness. For $(6,1,0)$, the six-chore bundle lies in its
owner's set $L_i$, so $M\setminus L_i$ has at most one element. By disjointness
the other two sets $L_j$ lie in $M\setminus L_i$, and both are nonempty, so they
are equal and intersect. For $(5,2,0)$ and $(4,3,0)$, both nonempty bundles lie
in their owners' sets $L_i$ and together cover $M$, so the nonempty set $L_k$ of
the agent with the empty bundle meets one of them. These are all occupancies
with an empty bundle. \qed
\end{proof}

By Lemma~\ref{lem:no-empty}, the possible occupancies of EFX allocations in this
case are $(3,2,2)$, $(3,3,1)$, $(4,2,1)$ and $(5,1,1)$, and each of them can be
the only occupancy of the EFX allocations of an instance.
Table~\ref{tab:occupancy} lists four matrices with pairwise disjoint sets of
cheapest chores, all of whose EFX allocations have the same occupancy.
The counts were obtained by going through all $2187$ allocations of each matrix;
the matrices and the script \file{examples/occupancy_examples.py} are in the
repository.

\begin{table}[htbp]
  \centering
  \caption{Four cost matrices (rows are agents $0,1,2$; columns are chores
  $0,\ldots,6$) for which all EFX allocations have the same occupancy.}
  \label{tab:occupancy}
  \small
  \setlength{\arraycolsep}{2.5pt}
  \begin{tabular}{@{}lcc@{}}
    \toprule
    Cost matrix & Occupancy &
    \begin{tabular}[b]{@{}c@{}}Number of\\ EFX allocations\end{tabular} \\
    \midrule
    $\begin{pmatrix}
      5 & 6 & 10 & 1 & 2 & 6 & 6\\
      1 & 5 & 6 & 2 & 3 & 4 & 4\\
      4 & 5 & 6 & 4 & 4 & 1 & 3
    \end{pmatrix}$ & $(3,2,2)$ & 64 \\[4ex]
    $\begin{pmatrix}
      42 & 943 & 37 & 16 & 34 & 1041 & 987\\
      3118 & 24680 & 6203 & 6302 & 6137 & 21793 & 34067\\
      6116 & 21754 & 3058 & 6174 & 3232 & 18609 & 30957
    \end{pmatrix}$ & $(3,3,1)$ & 21 \\[4ex]
    $\begin{pmatrix}
      8 & 330 & 6 & 0 & 159 & 191 & 346\\
      621 & 1071 & 507 & 4582 & 38 & 1502 & 1559\\
      0 & 525 & 0 & 8 & 24 & 513 & 1010
    \end{pmatrix}$ & $(4,2,1)$ & 8 \\[4ex]
    $\begin{pmatrix}
      4 & 1 & 2 & 1 & 0 & 9 & 53\\
      0 & 6 & 59 & 0 & 71 & 142 & 142\\
      3 & 0 & 0 & 62 & 3 & 71 & 71
    \end{pmatrix}$ & $(5,1,1)$ & 6 \\
    \bottomrule
  \end{tabular}
\end{table}

\end{document}